\documentclass[11pt]{article}

\def\ShowComment{True} % Switch comments  on or  off    

\usepackage[utf8]{inputenc}
\usepackage[T1]{fontenc}

\usepackage{authblk}

\usepackage{geometry}
\usepackage[bottom]{footmisc}
\usepackage{todonotes}
\usepackage{xspace}

\usepackage{comment}

\usepackage{nicefrac}
\usepackage[noend]{algpseudocode}

\usepackage{amssymb}
\usepackage{amsmath}
\usepackage{amsthm}
\usepackage{dsfont}
\usepackage{footnote}

\usepackage{array}
\usepackage{verbatim}
\usepackage{multirow}
\usepackage{xcolor}
\usepackage{nameref}
\definecolor{ForestGreen}{rgb}{0.1333,0.5451,0.1333}
\definecolor{DarkRed}{rgb}{0.65,0,0}
\definecolor{Red}{rgb}{1,0,0}
\usepackage[linktocpage=true,
pagebackref=true,colorlinks,
linkcolor=DarkRed,citecolor=ForestGreen,
bookmarks,bookmarksopen,bookmarksnumbered]
{hyperref}
\usepackage{cleveref}

\usepackage{thm-restate}
\usepackage[tableposition=bottom]{caption}
\usepackage[all]{hypcap}
\usepackage{subcaption}
\usepackage{framed}
\usepackage[framemethod=tikz]{mdframed}
\usepackage[skins]{tcolorbox}

\usepackage[ruled,vlined,linesnumbered]{algorithm2e}
\Crefname{algocf}{Algorithm}{Algorithms}

\makeatletter
\g@addto@macro{\maketitle}{\@thanks}
\makeatother

\newtheorem{theorem}{Theorem}[section]
\newtheorem{corollary}[theorem]{Corollary}
\newtheorem{lemma}[theorem]{Lemma}
\newtheorem{definition}[theorem]{Definition}

\newcommand{\p}{\textsc{P}}%
\newcommand{\poly}{\mathrm{poly}}

\renewcommand{\algorithmiccomment}[1]{\bgroup\hfill$\rhd$~#1\egroup}

\newcounter{note}[section]

\newcommand{\Otil}{\tilde{O}}

\global\long\def\poly{\mathrm{poly}}

\newenvironment{wrapper}[1]
{
	\begin{center}
		\begin{minipage}{\linewidth}
			\begin{mdframed}[hidealllines=true, backgroundcolor=gray!20, leftmargin=0cm,innerleftmargin=0.4cm,innerrightmargin=0.4cm,innertopmargin=0.4cm,innerbottommargin=0.4cm,roundcorner=10pt]
				#1}
			{\end{mdframed}
		\end{minipage}
	\end{center}
} 

\renewcommand{\paragraph}[1]{\medskip\noindent\textbf{#1}}

\ifdefined\ShowComment
\def\thatchaphol#1{\marginpar{$\leftarrow$\fbox{T}}\footnote{$\Rightarrow$~{\sf\textcolor{ForestGreen}{#1 --Thatchaphol}}}}

\def\sayan#1{\marginpar{$\leftarrow$\fbox{S}}\footnote{$\Rightarrow$~{\sf\textcolor{purple}{#1 --Sayan}}}}

\def\peter#1{\marginpar{$\leftarrow$\fbox{P}}\footnote{$\Rightarrow$~{\sf\textcolor{green}{#1 --Peter}}}}
\else
\def\thatchaphol#1{}
\def\sayan#1{}
\def\peter#1{}
\fi 
\usepackage{etoolbox}
\usepackage{tikz}
\usetikzlibrary{tikzmark}
\usetikzlibrary{calc}

\errorcontextlines\maxdimen

\newcommand{\ALGtikzmarkcolor}{black}% customise this, if you want
\newcommand{\ALGtikzmarkextraindent}{4pt}% customise this, if you want
\newcommand{\ALGtikzmarkverticaloffsetstart}{-.5ex}% customise this, if you want
\newcommand{\ALGtikzmarkverticaloffsetend}{-.5ex}% customise this, if you want
\makeatletter
\newcounter{ALG@tikzmark@tempcnta}

\newcommand\ALG@tikzmark@start{%
	\global\let\ALG@tikzmark@last\ALG@tikzmark@starttext%
	\expandafter\edef\csname ALG@tikzmark@\theALG@nested\endcsname{\theALG@tikzmark@tempcnta}%
	\tikzmark{ALG@tikzmark@start@\csname ALG@tikzmark@\theALG@nested\endcsname}%
	\addtocounter{ALG@tikzmark@tempcnta}{1}%
}

\def\ALG@tikzmark@starttext{start}
\newcommand\ALG@tikzmark@end{%
	\ifx\ALG@tikzmark@last\ALG@tikzmark@starttext
	\else
	\tikzmark{ALG@tikzmark@end@\csname ALG@tikzmark@\theALG@nested\endcsname}%
	\tikz[overlay,remember picture] \draw[\ALGtikzmarkcolor] let \p{S}=($(pic cs:ALG@tikzmark@start@\csname ALG@tikzmark@\theALG@nested\endcsname)+(\ALGtikzmarkextraindent,\ALGtikzmarkverticaloffsetstart)$), \p{E}=($(pic cs:ALG@tikzmark@end@\csname ALG@tikzmark@\theALG@nested\endcsname)+(\ALGtikzmarkextraindent,\ALGtikzmarkverticaloffsetend)$) in (\x{S},\y{S})--(\x{S},\y{E});%
	\fi
	\gdef\ALG@tikzmark@last{end}%
}

\apptocmd{\ALG@beginblock}{\ALG@tikzmark@start}{}{\errmessage{failed to patch}}
\pretocmd{\ALG@endblock}{\ALG@tikzmark@end}{}{\errmessage{failed to patch}}
\makeatother
\newcommand\abs[1]{\left\lvert #1 \right\rvert}

\newcommand\vol[0]{\mathrm{vol}}

\global\long\def\Otil{\tilde{O}}

\newcommand\Main{\textsc{Main}}

\title{Maximal \texorpdfstring{$k$}{k}-Edge-Connected Subgraphs in Almost-Linear Time\\ for Small \texorpdfstring{$k$}{k}}
\author[1]{Thatchaphol Saranurak\thanks{Supported by NSF CAREER grant 2238138.}}
\author[2]{Wuwei Yuan}
\affil[1]{University of Michigan}
\affil[2]{Institute for Interdisciplinary Information Sciences, Tsinghua University}

\date{}

\begin{document}

\maketitle
\pagenumbering{gobble}
\begin{abstract}
We give the first almost-linear time algorithm for computing the \emph{maximal $k$-edge-connected subgraphs} of an undirected unweighted graph for any constant $k$. More specifically, given an $n$-vertex $m$-edge graph $G=(V,E)$ and a number $k = \log^{o(1)}n$, we can deterministically compute in $O(m+n^{1+o(1)})$ time the unique vertex partition $\{V_{1},\dots,V_{z}\}$ such that, for every $i$, $V_{i}$ induces a $k$-edge-connected subgraph while every superset $V'_{i}\supset V_{i}$ does not. Previous algorithms with linear time work only when $k\le2$ {[}Tarjan SICOMP'72{]}, otherwise they all require $\Omega(m+n\sqrt{n})$ time even when $k=3$ {[}Chechik~et~al.~SODA'17; Forster~et~al.~SODA'20{]}.

Our algorithm also extends to the decremental graph setting; we can deterministically maintain the maximal $k$-edge-connected subgraphs of a graph undergoing edge deletions in $m^{1+o(1)}$ total update time. Our key idea is a reduction to the dynamic algorithm supporting pairwise $k$-edge-connectivity queries {[}Jin and Sun FOCS'20{]}. 
\end{abstract}

\newpage
\pagenumbering{arabic}

\section{Introduction}

We study the problem of efficiently computing the \emph{maximal $k$-edge-connected subgraphs}. Given an undirected unweighted graph $G=(V,E)$ with $n$ vertices and $m$ edges, we say that $G$ is \emph{$k$-edge-connected} if one needs to delete at least $k$ edges to disconnect $G$. The maximal $k$-edge-connected subgraphs of $G$ is a unique vertex partition $\{V_{1},\dots,V_{z}\}$ of $V$ such that, for every $i$, the induced subgraph $G[V_{i}]$ is $k$-edge-connected and there is no strict superset $V'_{i}\supset V_{i}$ where $G[V'_{i}]$ is $k$-edge-connected. 

This fundamental graph problem has been intensively studied. Since the 70's, Tarjan \cite{tarjan1972depth} showed an optimal $O(m)$-time algorithm when $k=2$. For larger $k$, the folklore recursive mincut algorithm takes $\Otil(mn)$ time\footnote{The algorithm computes a global minimum cut $(A,B)$ (using e.g.~Karger's algorithm \cite{karger2000minimum}) and return $\{V\}$ if the cut size of $(A,B)$ is at least $k$. Otherwise, recurse on both $G[A]$ and $G[B]$ and return the union of the answers of the two recursions.} and there have been significant efforts from the database community in devising faster heuristics \cite{yan2005mining,zhou2012finding,akiba2013linear,sun2016efficient,yuan2016efficient} but they all require $\Omega(mn)$ time in the worst case. Eventually in 2017, Chechik et al.~\cite{chechik2017faster} broke the $O(mn)$ bound to $\tilde{O}(m\sqrt{n}k^{O(k)})$ using a novel approach based on \emph{local }cut algorithms. Forster et al.~\cite{forster2020computing} then improved the local cut algorithm and gave a faster Monte Carlo randomized algorithm with $\tilde{O}(mk+n^{3/2}k^{3})$ running time. Very recently, Geogiadis et al.~\cite{georgiadis2022maximal} showed a deterministic algorithm with $\tilde{O}(m+n^{3/2}k^{8})$ time and also how to sparsify a graph to $O(nk\log n)$ edges while preserving maximal $k$-edge-connected subgraphs in $O(m)$ time. Thus, the factor $m$ in the running time of all algorithms can be improved to $O(nk\log n)$ while paying an $O(m)$ additive term. The $O(mn)$ bound has also been improved even in more general settings such as directed graphs and/or vertex connectivity \cite{henzinger2000computing,chechik2017faster,forster2020computing} as well as weighted undirected graphs \cite{nalam2023maximal}. Nonetheless, in the simplest setting of undirected unweighted graphs where $m=O(n)$ and $k=O(1)$, the $\Omega(n\sqrt{n})$ bound remains the state of the art since 2017.

Let us discuss the closely related problem called \emph{$k$-edge-connected components}. The goal of this problem is to compute the unique vertex partition $\{\hat{V}{}_{1},\dots,\hat{V}{}_{z'}\}$ of $V$ such that, each vertex pair $(s,t)$ is in the same part $\hat{V}{}_{i}$ iff the $(s,t)$-minimum cut in $G$ (not in $G[\hat{V}_{i}]$) is at least $k$. The partition of the maximal $k$-edge-connected subgraphs is always a refinement of the $k$-edge-connected components and the refinement can be strict. See \Cref{fig:fig1} for example. Very recently, the Gomory-Hu tree algorithm by Abboud et al.~\cite{abboud2022breaking} implies that $k$-edge-connected components can be computed in $m^{1+o(1)}$ time in undirected unweighted graphs. This algorithm, however, does not solve nor imply anything to our problem. See \Cref{sec:discuss} for a more detailed discussion.

\begin{figure}[htbp]
\centering
\includegraphics[width=1\textwidth]{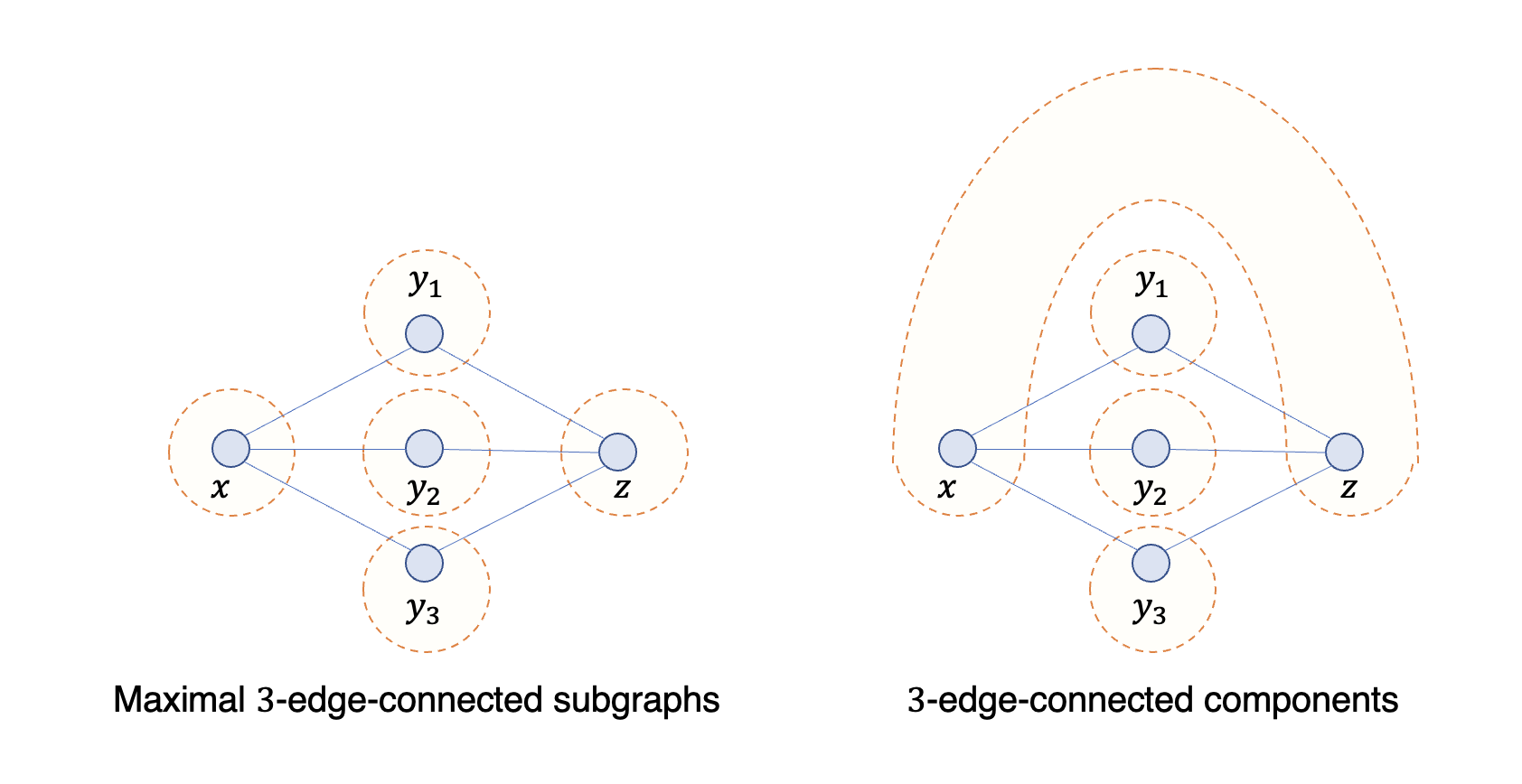}
\caption{A graph $G$ where its maximal $3$-edge-connected subgraphs are different from its $3$-edge-connected components.}
\label{fig:fig1}
\end{figure}

It is an intriguing question whether one can also obtain an almost-linear time algorithm for maximal $k$-edge-connected subgraphs, or there is a separation between these two closely related problems. 

\paragraph{Our results.}
In this paper, we show the first almost-linear time algorithm when $k=\log^{o(1)}n$, answering the above question affirmatively at least for small $k$.

\begin{wrapper}
\begin{theorem}
\label{theorem:deterministic_kecs}
\label{thm:main}There is a deterministic algorithm that, given an undirected unweighted graph $G$ with $n$ vertices and $m$ edges, computes the maximal $k$-edge-connected subgraphs of $G$ in $O(m+n^{1+o(1)})$ time for any $k=\log^{o(1)}n$.
\end{theorem}
\end{wrapper}

Our techniques naturally extend to the decremental graph setting.

\begin{wrapper}
\begin{theorem}
\label{thm:main dyn}There is a deterministic algorithm that, given an undirected unweighted graph $G$ with $n$ vertices and $m$ edges undergoing a sequence of edge deletions, maintains the maximal $k$-edge-connected subgraphs of $G$ in $m^{1+o(1)}$ total update time for any $k=\log^{o(1)}n$.
\end{theorem}
\end{wrapper}

Dynamic algorithms for maximal $k$-edge-connected subgraphs were recently studied in \cite{georgiadis2022maximal}. For comparison, their algorithm can handle both edge insertions and deletions but require $O(n\sqrt{n}\log n)$ worst-case update time, which is significantly slower than our  $m^{o(1)}$ amortized update time. When $k=3$, they also gave an algorithm that handles edge insertions only using $\Otil(n^{2})$ total update time. 

\paragraph{Previous Approaches and Our Techniques.}
Our approach diverges significantly from the local-cut-based approach in \cite{chechik2017faster,forster2020computing}. In these previous approaches, they call the local cut subroutine $\Omega(n)$ times and each call takes $\Omega(\sqrt{n})$ time. Hence, their running time is at least $\Omega(n\sqrt{n})$ and this seems inherent without significant modification. 
% how In fact, we explain why this approach has a technical barrier that leads to $\Omega(n\sqrt{n})$ time in Appendix \textbf{ADDTHIS}.\wuwei{This explanation does not seem to make too much sense. One can estimate the volume using other methods.} 
Recently, \cite{georgiadis2022maximal} took a different approach. Their $\tilde{O}(m+n^{3/2}k^{8})$-time algorithm  efficiently implements the folklore recursive mincut algorithm by feeding $O(nk)$ updates to the dynamic minimum cut algorithm by Thorup \cite{thorup2007fully}. However, since Thorup's algorithm has $\Omega(\sqrt{n})$ update time, the final running time of \cite{georgiadis2022maximal} is at least $\Omega(n\sqrt{n})$ as well.

Our algorithm is similar to \cite{georgiadis2022maximal} in spirit but is much more efficient. We instead apply the dynamic $k$-edge connectivity algorithms by Jin and Sun \cite{jin2022fully} that takes only $n^{o(1)}$ update time when $k=\log^{o(1)}n$. Our reduction is more complicated than the reduction in \cite{georgiadis2022maximal} to dynamic minimum cut because the data structure by \cite{jin2022fully} only supports pairwise $k$-edge connectivity queries, not a global minimum cut. Nonetheless, we show that $\Otil(nk)$ updates and queries to this ``weaker'' data structure also suffice. 

Our approach is quite generic. Our algorithm is carefully designed without the need to check if the graph for which the recursive call is made is $k$-edge-connected. This allows us to extend our algorithm to the dynamic case. 

\paragraph{Organization.}
We give preliminaries in \Cref{sec:prelim}. Then, we prove \Cref{thm:main} and \Cref{thm:main dyn} in \Cref{sec:static} and \Cref{sec:dynamic}, respectively.

\label{sec:typesetting-summary}
\section{Preliminaries}
\label{sec:prelim}

    Let $G = (V, E)$ be an \emph{unweighted undirected} graph. Let $n = \abs{V}$ and $m = \abs{E}$, and assume $m = \poly(n)$ and $k=\log^{o(1)}n$. For any $S, T \subseteq V$, let $E(S, T) = \{(u, v) \in E \mid u \in S, v \in T\}$. For every vertex $u$, the degree of $u$ is $\deg(u) = \abs{\{(u, v) \mid (u, v) \in E\}}$. For every subset of vertices $S \subseteq V$, the volume of $S$ is  $\vol(S) = \sum_{u \in S} \deg(u)$. Denote $G[S]$ as the induced graph of $G$ on a subset of vertices $S\subseteq V$. 

    Two vertices $s$ and $t$ are \emph{$k$-edge-connected} in $G$ if one needs to delete at least $k$ edges to disconnect $s$ and $t$ in $G$.     
    A vertex set $S$ is \emph{$k$-edge-connected} if every pair of vertices in $S$ is $k$-edge-connected. 
    We use the convention that $S$ is $k$-edge-connected when $|S|=1$.
    We say that a graph $G=(V,E)$ is $k$-edge-connected if $V$ is $k$-edge-connected. 
    A \emph{$k$-edge-connected component} is an inclusion-maximal vertex set $S$ such that $S$ is $k$-edge-connected. 
    A whole vertex set can always be partitioned into $k$-edge-connected components. 
    We use $kECC(u)$ to denote the unique $k$-edge-connected component containing $u$. 
    Note that a $k$-edge-connected component may not induce a connected graph when $k > 2$.  
    A vertex set $S$ is a \emph{$k$-cut} if $|E(S,V\setminus S)| < k$. Note, however, we also count the whole vertex set $V$ as a trivial $k$-cut.

We will crucially exploit the following dynamic algorithm in our paper.

\begin{theorem}[Dynamic pairwise $k$-edge connectivity \cite{jin2022fully}] \label{theorem:dynamic_kec}
    There is a deterministic algorithm that maintains a graph $G$ with $n$ vertices undergoing edge insertions and deletions using $n^{o(1)}$ update time and, given any vertex pair $(s,t)$, reports whether $s$ and $t$ are $k$-edge-connected in the current graph $G$ in $n^{o(1)}$ time where $k = \log^{o(1)}n$.
\end{theorem}

For the maximal $k$-edge-connected subgraph problem, we can assume that the graph is sparse using the \emph{forest decomposition}.

\begin{definition}[Forest decomposition \cite{nagamochi2008encyclopedia}]
    A $t$-forest decomposition of a graph $G$ is a collection of forests $F_1, \ldots, F_t$, such that $F_i$ is a spanning forest of $G \setminus \bigcup_{j=1}^{i-1}F_j$, for every $1 \leq i \leq t$.
\end{definition}

\begin{theorem}[Lemma 8.3 of \cite{georgiadis2022maximal}] \label{theorem:certificate}
     Any $O(k \log n)$-forest decomposition of a graph has the same maximal $k$-edge-connected subgraphs as the original graph.
     Moreover, there is an algorithm for constructing such a $O(k \log n)$-forest decomposition in $O(m)$ time.
\end{theorem}

\section{The Static Algorithm}
\label{sec:static}

In this section, we prove our main result, \Cref{thm:main}.
The key idea is the following reduction:

\begin{lemma} \label{lemma:deterministic_kecs}
Suppose there is a deterministic decremental algorithm supporting pairwise $k$-edge-connectivity that has $t_p \cdot m$ total preprocessing and update time on an initial graph with $n$ vertices and $m$ edges and query time $t_q$. 

Then, there is a deterministic algorithm for computing the maximal $k$-edge-connected subgraphs in $O(m + (t_p + t_q) \cdot kn \log^2 n)$ time.
    
\end{lemma}

By plugging in \cref{theorem:dynamic_kec}, we get \cref{theorem:deterministic_kecs}. The rest of this section is for proving \Cref{lemma:deterministic_kecs}.
Throughout this section, we let $t_q$ denote the query time of the decremental pairwise $k$-edge connectivity data structure that \Cref{lemma:deterministic_kecs} assumes.

Recall again that, for any vertex $u$, $u$'s $k$-edge-connected component, $kECC(u)$, might not induce a connected graph.
The first tool for proving \Cref{lemma:deterministic_kecs} is a ``local'' algorithm for finding a connected component of $G[kECC(u)]$.

\begin{lemma}
    Given a graph $G$ and a vertex $u$, there is a deterministic algorithm for finding the connected component $U$ containing $u$ of $G[kECC(u)]$ in $O(t_q \cdot \vol(U))$ time.
\end{lemma}

\begin{proof}
    We run BFS from $u$ to explore every vertex in the connected component $U$ containing $u$ of $kECC(u)$. During the BFS process, we only visit the vertices in $kECC(u)$ by checking if the newly found vertex is $k$-edge-connected to $u$. Since each edge incident to $U$ is visited at most twice, the total running time is $O(t_q \cdot \vol(U))$.
\end{proof}

Below, we describe the algorithm for \Cref{lemma:deterministic_kecs} in \Cref{algo:deterministic_kecs}
and then give the analysis.

\begin{algorithm}
    \caption{$\Main(G, L)$: compute the maximal $k$-edge-connected subgraphs}
    \label{algo:deterministic_kecs}
    \DontPrintSemicolon
    \SetAlgoLined
    \KwIn{An undirected connected graph $G = (V, E)$, and a list of vertices $L$ (initially $L = V$). Note that the parameters are passed by value.}
    \KwOut{The maximal $k$-edge-connected subgraphs of $G$.}
    
    $S \gets \varnothing$.\\
    \While{$\abs{L}>1$}{
        Choose an arbitrary pair $(u,v)\in L$.\\
        \eIf{$u$ and $v$ are $k$-edge-connected in $G$}{
            $L \gets L \setminus \{v\}$.\label{line:update L connected}\\
        }{
            Simultaneously compute the $u$'s connected component of $G[kECC(u)]$ and the $v$'s connected component of $G[kECC(v)]$, until the one with the smaller volume (denoted by $U$) is found.\label{line:find U}\\
            $S\gets S\cup\Main(G[U], U)$.\label{line:recurse}
            \label{line:update L cut}\\
            $G\gets G\setminus U$.\label{line:delete}\\
            $L\gets(L\setminus U)\cup\{w\notin U\mid(x,w)\in E(U,V(G)\setminus U)\}$.
        }
    }
    $S \gets S \cup \{V(G)\}$.\\
    \Return $S$.

\end{algorithm}

\paragraph{Correctness.}
We start with the following structural lemma.
\begin{lemma}
    [Lemma 5.6 of \cite{chechik2017faster}]\label{lem:key}Let $T$ be a $k$-cut in $G[C]$ for some vertex set $C$. Then, either
    \begin{itemize}
    \item $T$ is a $k$-cut in $G$ as well, or
    \item $T$ contains an endpoint of $E(C,V(G)\setminus C)$. 
    \end{itemize}
\end{lemma}

Next, the crucial observation of our algorithm is captured by the following invariant.
\begin{lemma}
    \label{lem:invariant}At any step of \Cref{algo:deterministic_kecs}, every $k$-cut $T$ in $G$ is such that $T\cap L\neq\emptyset$.
\end{lemma}

\begin{proof}
    The base case is trivial because $L\gets V$ initially. Next, we prove the inductive step. $L$ can change in Line \ref{line:update L connected} or Line \ref{line:update L cut}.
    
    In the first case, the algorithm finds that $u$ and $v$ are $k$-edge-connected and removes $v$ from $L$. For any $k$-cut $T$ where $v\in T$, an important observation is that $kECC(v)\subseteq T$ as well. But $kECC(u)=kECC(v)$ and so $u\in T$ too. So the invariant still holds even after removing $v$ from $L$.
    
    In the second case, the algorithm removes $U$ from $G$. Let us denote $G'=G\setminus U$. Since the algorithm adds the endpoints of cut edges crossing $U$ to $L$, it suffices to consider a $k$-cut $T$ in $G'$ that is disjoint from the endpoints of the cut edges of $U$. By \Cref{lem:key}, $T$ was a $k$-cut in $G$. Since the changes in $L$ occur only at $U$ and neighbors of $U$, while $T$ is disjoint from both $U$ and all neighbors of $U$, we have $T\cap L\neq\emptyset$ by the induction hypothesis. 
    \end{proof}
\begin{corollary}\label{cor:stop k connected}
    When $|L|=1$, then $G$ is $k$-edge-connected. 
\end{corollary}
\begin{proof}
    Otherwise, there is a partition $(A,B)$ of $V$ where $|E(A,B)|<k$. So both $A$ and $B$ are $k$-cuts in $G$. By \Cref{lem:invariant}, $A\cap L\neq\emptyset$ and $B\cap L\neq\emptyset$ which contradicts that $|L|=1$. 
\end{proof}

We are ready to conclude the correctness of \Cref{algo:deterministic_kecs}. At a high level, the algorithm finds the set $U$ and ``cuts along'' $U$ at Lines \ref{line:find U}. Then, on one hand, recurse on $U$ at Line \ref{line:recurse} and, on the other hand, continue on $V(G)\setminus U$. We say that the cut edges $E(U,V(G)\setminus U)$ are ``deleted''. 

Now, since $U$ is the connected component of $G[kECC(u)]$ for some vertex $u$. We have that, for every edge $(x,y)\in E(U,V(G)\setminus U)$, the pair $x$ and $y$ are not $k$-edge-connected in $G$. In particular, $x$ and $y$ are not $k$-edge-connected in $G[V']$ for every $V'\subseteq V$. 

Thus, the algorithm never deletes edges inside any maximal $k$-edge-connected subgraph $V_{i}$. Since the algorithm stops only when the remaining graph is $k$-edge-connected, the algorithm indeed returns the maximal $k$-edge-connected subgraphs of the whole graph.

\paragraph{Running Time.}

Consider the time spent on each recursive call. 
Let $G'$ be the graph for which the recursive call is made and $m' = \vol(G')$.
Every vertex is inserted to $L$ initially or as an endpoint of some removed edge, so the total number of vertices added to $L$ is $O(m')$.
In each iteration, either we remove a vertex from $L$, or remove a subgraph from $G$. Hence we check pairwise $k$-edge-connectivity $O(m')$ times, so the running time of checking pairwise $k$-edge-connectivity is $O(t_q \cdot m')$.
For the time of finding connected components of $k$-edge-connected components, since we spend $O(t_q \cdot \vol(U))$ time to find some $U$ and remove $U$ from $G$, the total cost is $O(t_q \cdot m')$.
Plus, initializing the dynamic pairwise $k$-edge connectivity algorithm on $G'$ takes $O(t_p \cdot m')$ time.
Thus the total running time of each recursive call is $O((t_p + t_q) \cdot m')$.

For the recursion depth, since each $U$ found has the smaller volume of the two, $\vol(U) \leq m'/2$. Hence the recursion depth is $O(\log m_0)$, where $n_0$ and $m_0$ are the numbers of vertices and edges of the initial graph.
Thus the total running time of \Cref{algo:deterministic_kecs} is $O((t_p + t_q) \cdot m_0 \log n_0)$.

By applying \cref{theorem:certificate} to the initial graph $G$ and invoking \Cref{algo:deterministic_kecs}  on the resulting graph, the number of edges in the resulting graph is $O(kn_0 \log n_0)$, so the running time is improved to 
$O(m_0 + (t_p + t_q) \cdot kn_0 \log^2 n_0)$.
This completes the proof of \Cref{lemma:deterministic_kecs}.

\section{The Decremental Algorithm}
\label{sec:dynamic}

Our static algorithm can be naturally extended to a decremental dynamic algorithm. To prove \Cref{thm:main dyn}, we prove the following reduction. By combining \Cref{lemma:dynamic_kecs} and \Cref{theorem:dynamic_kec}, we are done.

\begin{lemma} \label{lemma:dynamic_kecs}
  Suppose there is a deterministic decremental algorithm supporting pairwise $k$-edge-connectivity that has $t_p \cdot m$ total preprocessing and update time on an initial graph with $n$ vertices and $m$ edges and query time $t_q$. 
    
    Then there is a deterministic decremental dynamic algorithm for maintaining the maximal $k$-edge-connected subgraphs on an undirected graph of $n$ vertices and $m$ edges with
    $O((t_p + t_q) \cdot m \log n)$
    total preprocessing and update time, and $O(1)$ query time.

\end{lemma}
    
The algorithm for \Cref{lemma:dynamic_kecs} as is follows. 
First, we preprocess the initial graph $G_0$ using \Cref{algo:deterministic_kecs} and obtain the maximal $k$-edge-connected subgraphs $\{V_1,\dots,V_z\}$ of $G_0$.

Next, given an edge $e$ to be deleted, if $e$ is in a maximal $k$-edge-connected subgraph $V_i$ of $G$, then we invoke $\textsc{Update}(G[V_i], e)$ and update the set of the maximal $k$-edge-connected subgraphs of $G$; otherwise we ignore $e$.
The subroutine $\textsc{Update}(H, e)$ is described in \Cref{algo:update_kecs}.

\begin{algorithm}
\caption{\textsc{Update}(H, e)}
\label{algo:update_kecs}
\DontPrintSemicolon
\SetAlgoLined
\KwIn{A $k$-edge-connected subgraph $H$ and an edge $e = (x, y) \in H$ to be deleted.}
\KwOut{The $k$-edge-connected subgraphs of $H$ after deletion.}
$H \gets H \setminus \{(x, y)\}$.\\
\Return $\Main(H, \{x, y\})$.\\

\end{algorithm}

\paragraph{Correctness.}

Let $H = (V', E')$ be the maximal $k$-edge-connected subgraph containing edge $(x, y)$ before deletion. 
It suffices to prove that \Cref{lem:invariant} holds when we invoke \Cref{algo:deterministic_kecs}. Suppose there is a $k$-cut $C$ in $H \setminus \{(x, y)\}$ such that $C \cap \{x, y\} = \emptyset$, then $C$ is also a $k$-cut in $H$, a contradiction.
Hence the correctness follows from the correctness of \Cref{algo:deterministic_kecs}.

\paragraph{Running Time.}

In the case that $H \setminus \{(x, y)\}$ is still $k$-edge-connected, the running time is $t_q$. We charge this time $t_q$ to the deleted edge $(x, y)$.

Otherwise, consider the time spend on each recursive call of $\Main$. Assume that the total volume of the subgraphs removed and passed to another recursive call in a recursive call is $\nu$. The total number of vertices added to $L$ is $O(\nu)$. 
In each iteration, we either remove a vertex from $L$ or remove a subgraph. Hence we check pairwise $k$-edge-connectivity $O(\nu)$ times, so the running time of checking pairwise $k$-edge-connectivity is $O(t_q \cdot \nu)$.
Since we spend $O(t_q \cdot \vol(U))$ time to find $U$, the total cost is $O(t_q \cdot \nu)$. 
Plus, it takes $O((t_p + t_q)\cdot m')$ time to initialize the dynamic pairwise $k$-edge connectivity algorithm and check pairwise $k$-edge-connectivity on a graph $H'$ with $m'$ edges for the first time we invoke $\Main$ on $H'$.
Also, removing all edges from $H'$ takes $t_p \cdot m'$ time.
We charge $O(t_p + t_q)$ to each of the removed edges in each recursive call.

The recursion depth is $O(\log m_0)$ by \Cref{lemma:deterministic_kecs}, where $n_0$ and $m_0$ are the numbers of vertices and edges of the initial graph. Hence each edge will be charged $O(\log m_0)$ times, so the total preprocessing and update time is $O((t_p + t_q ) \cdot m_0 \log n_0)$.

\appendix

\section{Relationship with \texorpdfstring{$k$}{k}-Edge-Connected Components}
\label{sec:discuss}
The reason why a subroutine for computing $k$-edge-connected components is not useful for computing maximal $k$-edge-connected subgraphs is as follows. 
Given a graph $G = (V,E)$, we can artificially create a supergraph $G' = (V' \supseteq V, E' \supseteq E)$ where the whole set $V$ is $k$-edge-connected, but the maximal $k$-edge-connected subgraphs of $G'$ will reveal the maximal $k$-edge-connected subgraphs of $G$. So given a subroutine for computing the $k$-edge-connected components of $G'$, we know nothing about the maximal $k$-edge-connected subgraphs of $G$. The construction of $G'$ is as follows. 

First, we set $G' \gets G$. Assume $V = \{1, 2, \ldots, n\}$. For every $1 \leq i < n$, add $k$ parallel dummy length-$2$ paths $(i, d_{i,1}, i+1) , \ldots , (i, d_{i,k}, i+1)$. Thus $i$ and $i+1$ are $k$-edge-connected, so $V$ is $k$-edge-connected at the end. 
When we compute the maximal $k$-edge-connected subgraphs of $G'$, we know that we will first remove all dummy vertices $d_{i,j}$ because they all have degree $2$ (assuming that $k>2$). We will obtain $G$ and so we will obtain the maximal $k$-edge-connected subgraphs of $G$ from this process.

\newpage
\bibliographystyle{alpha}
\bibliography{main.bib}

\end{document}